\documentclass{amsart} %{amsart}%{article}%{elsart}
\usepackage{geometry}
\usepackage{amssymb, amsmath, sidecap}
\usepackage{amsfonts}
\usepackage{float}
\usepackage{graphics}
\usepackage{graphicx}
\usepackage{subfigure}
\usepackage{verbatim}
\usepackage{textcomp}
\usepackage{caption2}
\usepackage{ textcomp }
\usepackage{ dsfont }

\def\bt{\begin{thm}}
\def\et{\end{thm}}
\def\bl{\begin{lem}}
\def\el{\end{lem}}
\def\bd{\begin{defi}}
\def\ed{\end{defi}}
\def\bc{\begin{cor}}
\def\ec{\end{cor}} 
\def\bp{\begin{proof}}
\def\ep{\end{proof}}
\def\br{\begin{rem}}
\def\er{\end{rem}}
\def\pr{\text{Pr}}
\def\div{\text{div}}
\def\bt{\begin{thm}}
\def\et{\end{thm}}
\def\bl{\begin{lem}}
\def\el{\end{lem}}
\def\bd{\begin{defi}}
\def\ed{\end{defi}}
\def\bc{\begin{cor}}
\def\ec{\end{cor}}
\def\bp{\begin{proof}}
\def\ep{\end{proof}}
\def\br{\begin{rem}}
\def\er{\end{rem}}

\newtheorem{thm}{Theorem}[section]
\newtheorem{lem}{Lemma}[section]
\newtheorem{defi}{Definition}[section]

\newtheorem{rem}{Remark}[section]
\newtheorem{cor}{Corollary}[section]

\newcommand{\mathsym}[1]{{}}%% AMS-LaTeX Created by Wolfram Mathematica 7.0 : www.wolfram.com
\newcommand{\unicode}[1]{{}}%% AMS-LaTeX Created by Wolfram Mathematica 7.0 : www.wolfram.com

\newcommand{\pd}[2]{\frac{\textstyle \partial #1}{\textstyle \partial #2}}

\def\la{\label}
\def\curl{\text{curl}}

\newcommand{\fr}[2]{\frac{#1}{#2}}

\newcommand{\de}{\partial}

\numberwithin{equation}{section}
\graphicspath{{./figures/}}
\title{Remarks on the Rayleigh-B\'enard Convection on Spherical Shells }
\author[Wang]{Shouhong Wang}
\address[Wang]{Department of Mathematics,
Indiana University, Bloomington, IN 47405}
\email{showang@indiana.edu, http://www.indiana.edu/\texttildelow fluid}

\author[Yang]{Ping Yang}
\address[Yang]{Department of Mathematics,
Indiana University, Bloomington, IN 47405}
\email{pingyang@umail.iu.edu}

\thanks{The work  was supported in part by the
US Office of Naval Research and by the US National Science Foundation.}

\begin{document}

\begin{abstract}The main objective of this article is to study the effect of spherical geometry on dynamic transitions and pattern formation for the Rayleigh-B\'enard convection. The study is mainly motivated  by the importance of spherical geometry and convection in geophysical flows. It is shown in particular that the system always undergoes a continuous (Type-I) transition  to a $2l_c$-dimensional sphere $S^{2l_c}$, where $l_c$  is the critical wave length corresponding to the critical Rayleigh number. Furthermore, it has shown in \cite{MW09c} that it is critical to add nonisotropic turbulent friction terms in the momentum equation to capture the large-scale atmospheric and oceanic circulation patterns. We show in particular that the system with turbulent friction terms added undergoes the same type of dynamic transition, and obtain an explicit formula linking the critical wave number (pattern selection), the aspect ratio, and the ratio between the horizontal and vertical turbulent friction coefficients.
\end{abstract}

\maketitle
\tableofcontents

\section{Introduction}
The main objective of this article is to study the  B\'enard convection of fluids on a spherical shells. The main motivations are two-fold. First, convection plays an important role in the large scale atmospheric and oceanic circulations. The phenomena of the atmosphere and ocean are extremely rich in their organization and complexity, and  involve a broad range of  temporal and spatial scales \cite{charney48}. To avoid to deal with the atmosphere and ocean in all their complexity, the main focus of this article is to study the B\'enard convection on spherical shells to capture the main effects caused by the spherical geometry. 
Second,  over the years,  the Rayleigh-B\'enard convection problem, together with the Taylor problem, has
become one of the paradigms for studying nonequilibrium phase transitions and pattern formation in nonlinear sciences; see among others  Busse \cite{busse1978non}, Chandrasekhar \cite{chandrasekhar}, Cross \cite{Cross1993a}, Ma and Wang \cite{ptd}, and the references therein.

For the Euclidean domain case, a systematic study on dynamic transition of the Rayleigh-B\'enard convection  is carried out in Ma and Wang \cite{MW04d,MW07a}. The main result in this direction is that the system always undergoes a Type-I (continuous) transition as the instability driving mechanism, namely Rayleigh number, crosses a critical threshold $R_c$, thanks to the symmetry of the linear operator, properties of the nonlinearity and asymptotic stability of the basic state at the critical threshold. Moreover, the system has a bifurcated attractor which is an $(m-1)$--dimensional homological sphere where $m$ is the number of critical eigenvalues of the linear operator. 

In this article, we consider a fluid which occupies a layer of spherical shell modeled by the  Boussinesq equations. Namely, the density function is treated as a constant except in the equation of state and in the buoyancy term. For simplicity, the spatial geometry of the domain is taken as the product of a two-dimensional sphere $S^2_a$ with radius $a$ and an interval: $S^2_a \times (0, h)$, where $h$  is the height of the fluid layer. This is mainly motivated by the fact that aspect ratio for the large scale atmosphere and ocean is small; see among others \cite{LTW92a, LTW92b, pedlosky87}. 

As in the Euclidean case, in the spherical shell case, we show that the system is always undergoes a Type-I (continuous) dynamic transition to a $2l_c$-dimensional homological sphere $\Sigma_R$, where $l_c$ is the critical wave number determined by \footnote{The integer $l_c$ is to make the left-hand side the closet to the right-hand side.} 
\begin{equation}\la{wavenumber}
\left[\frac{h}{a}\right]^2l_c(l_c+1) =\frac{\pi^2}{2}.
\end{equation}

An important part of this article is to show that $2l_c$-dimensional homological sphere $\Sigma_R$  is in fact {\it homeomorphic} to a $2l_c$-dimensional sphere $S^{2l_c}$. This is achieved by carrying out the reduction of the original Boussinesq equations to the center manifold generated by the first $2l_c+1$  modes. The main difficulty comes then from the nontrivial nonlinear interactions of the spherical harmonics--both scalar and vectorial. 

The remaining part of this article  is devoted to the introduction of  the turbulent  friction terms in the Boussinesq equations. This is crucial for capturing  the large scale convection patterns in the large scale geophysical flows, as demonstrated in thermohaline circulation studies by Ma and Wang \cite{MW09c}. The friction terms are added based on the convection scale considerations as introduced in \cite{MW09c}, and are given by $(\sigma_0 u, \sigma_1w)$; see (\ref{e3}). Under the assumption that 
$$ 1 \ll \sigma_0 \ll \sigma_1,$$
the corresponding critical wave number $l_c$ is approximately given by 
\begin{equation}\la{wavenumber}
\left[\frac{h}{a}\right]^2 l_c(l_c+1) \sim \frac{\pi^2}{2} \left[\frac{\sigma_0}{\sigma_1}\right]^{1/2}.
\end{equation}
As we know, for the large-scale atmospheric and oceanic circulations, the critical wave number $l_c$  is often relatively small, e.g. $l_c=6$ for the global Walker circulation over the tropics. Hence the small aspect ratio $h/a$  is then balanced by the smallness of the ratio $\sigma_0/\sigma_1$ between the horizontal and vertical turbulent friction coefficients, leading to moderated size of the critical wave numbers. 

It is worth mentioning that with the framework given in this article, one can study the humidity/salinity effects of the large scale atmospheric and oceanic flows, as well as the rotational effect (Coriolis force) of the earth. We shall report this elsewhere. 

The article is organized as follows. Section 2 introduces the Boussinesq equations on the spherical shell, 
and the linear stability analysis is carried in Section 3. Sections 3 and 4 are on dynamic transitions and the structure of the bifurcated attractor. Section 6 demonstrates the need to introduction of convection scales, and studies the dynamics transitions for the Boussinesq equations with the added friction terms. 

\section{Boussinesq Equations}
We start with the prototype of problem, the B\'enard convection problem  for a layer of  fluid in a spherical shell. We consider the case where the fluid can be modeled by the Boussinesq equations. Namely, the density is considered as constant except in the buoyancy term and in the equation of state.  The equations are given as follows:
\begin{equation}\label{e1}
\begin{aligned}
&u_t+(u\cdot{\nabla_3})u-\nu\Delta_3{u}+\rho_0^{-1}\nabla_3p=-g\overrightarrow{k}(1-\alpha(T-\overline{T_0})),\\
&T_t+(u\cdot\nabla_3){T}-\kappa\Delta_3T=0,\\
&\div_3u=0.
\end{aligned}
\end{equation}
Here the unknown functions are velocity field $u=(u_1,u_2,u_3)$, the temperature function $T$, and the pressure function $p$. Other parameters include: $\nu$ is the kinematic viscosity, $\overrightarrow{k}=(0,0,1)$ is the unit vector in the vertical direction, g is the gravity constant, $\alpha$ is coefficient of thermal expansion of the fluid, $\overline{T_0}$ is temperature on the lower surface, $\rho_0$ is the density,$\kappa$ is he thermal diffusivity.For simplicity, we consider the spherical domain as 
$$\Omega_d =S_a^2\times(0,h),$$
where $S_a^2$ is the two dimensional sphere with radius $\alpha$, and $h$ is the height of the fluid layer. The motion for this type of geometry is due to the study of large scale geophysical fluid flows, where the aspect ratio, and the thin layer of spherical shell can be approximated by $\Omega_d$ given here.

It's easy to get the basic solution:
\[\begin{aligned}
&(u,T)=(0,\overline{T_0}-\beta{x_3}),\\
&p=p_0-\rho_0g(x_3+\frac{\alpha\beta}{2}x_3^2).
\end{aligned}
\]
where $$\beta=\frac{\overline{T_0}-\overline{T_1}}{h},\ r=\frac{a}{h},$$
and $\overline{T_1}$ is temperature on the upper surface.

Mathematically speaking, it is convenient to introduce the nondimensional form of the equation. For this purpose, let
\[\begin{aligned}
&(x,t)=(hx',\frac{h^2t'}{\kappa}),\\
&(u,T)=(\frac{\kappa u,}{h},\beta{h}\frac{T'}{\sqrt{R}}+\overline{T_0}-\beta{h}x_3'),\\
&p=\rho_0\frac{\kappa^2p'}{h^2}+p_0-g\rho_0(hx_3'+\frac{\alpha\beta{h^2}}{2}(x_3')^2),\\
\end{aligned}\]
In addition, the Rayleigh number $R$ and the Prandtl number $\pr$  are defined by
$$R=\frac{g\alpha\beta}{\kappa\nu},\  \pr=\fr{\nu}{\kappa}.$$
The nondimensional domain is $$\Omega= S_r^2\times(0,1),$$
and it's reasonable to use the spherical coordinates. Let $(u,w)$ be the 3 dimensional velocity, u be the horizontal velocity field, and w be the vertical velocity. Also the following notations $\div, \nabla$ are differential operators in the horizontal direction. Then the system (\ref{e1}) can be written as:
 \begin{equation}\label{e22}
 \begin{aligned}
 &\frac{1}{\pr}(u_t+\nabla_uu+w\frac{\partial{u}}{\partial{z}}+\nabla{p})-(\Delta+\frac{\partial^2}{\partial{z}^2})u=0,\\
 &\frac{1}{\pr}(w_t+\nabla_uw+w\frac{\partial{w}}{\partial{z}}+\frac{\partial{p}}{\partial{z}})-(\Delta+\frac{\partial^2}{\partial{z}^2})w-\sqrt{R}T=0,\\
 &T_t+\nabla_uT+w\frac{\partial{w}}{\partial{z}}-\sqrt{R}w-(\Delta+\frac{\partial^2}{\partial{z}^2})T=0,\\
 &\div{u}+\frac{\partial{w}}{\partial{z}}=0,
 \end{aligned}
 \end{equation}
supplemented with the following free-slip boundary condition:
 \begin{equation}
 \label{e23}w=0, T=0,  \frac{\partial{u}}{\partial{z}}=0\ at \ z=0,1.
 \end{equation}
It is worth mentioning that we can also study other physically sound boundary conditions such as the non-slip boundary  or the periodic boundary conditions.

 The following are some differential operators at spherical coordinates on $S_r^2$:
 \begin{equation*}
 \begin{aligned}
 &\nabla T=\fr{1}{r}\pd{ f}{\theta}e_{\theta}+\fr{1}{r\sin\theta}\pd{f}{\varphi}e_{\varphi},\\
 &\curl T=\fr{1}{r\sin\theta}\pd{f}{\varphi}e_{\theta}-\fr{1}{r}\pd{ f}{\theta}e_{\varphi},\\
 &\nabla_uv=\frac{1}{r}(u_\theta\frac{\partial{v_\theta}}{\partial\theta}+\frac{u_\varphi}{\sin\theta}\frac{\partial{v_\theta}}{\partial\varphi}-u_\varphi{v}_\varphi{\cot\theta})e_\theta+\frac{1}{r}(u_\theta\frac{\partial{v_\varphi}}{\partial\theta}+\frac{u_\varphi}{\sin\theta}\frac{\partial{v_\varphi}}{\partial\varphi}+u_\varphi{v}_\theta{\cot\theta})e_\varphi,\\
&\nabla_uT=\frac{1}{r}(u_\theta\frac{\partial{T}}{\partial\theta}+\frac{u_\varphi}{\sin\theta}\frac{\partial{T}}{\partial\varphi}),\\
&\Delta{T}=\frac{1}{r^2}\frac{1}{\sin\theta}\left[\frac{\partial}{\partial\theta}(\sin\theta\frac{\partial{T}}{\partial\theta})+\frac{1}{\sin\theta}\frac{\partial^2T}{\partial\varphi^2}\right],\\
&\Delta{u}=(\Delta{u_\theta}-\frac{2\cos\theta}{r^2\sin^2\theta}\frac{\partial{u_\varphi}}{\partial\varphi}-\frac{u_\theta}{r^2\sin^2\theta})e_\theta+(\Delta{u_\varphi}+\frac{2\cos\theta}{r^2\sin^2\theta}\frac{\partial{u_\theta}}{\partial\varphi}-\frac{u_\varphi}{r^2\sin^2\theta})e_\varphi.
\end{aligned}
\end{equation*}
where $u$, $v$ are vectors, and $T$ is scale.

For the problem (\ref{e22}) with (\ref{e23}), we define the function spaces as 
\begin{equation}\label{e5}
\begin{aligned}
&H=\{(u,w,T)\in L^2(\Omega)^4 |\  div u=0; (w,T)=0,\ at \ z=0,1\},\\
&H_1=\{ (u,w,T)\in H^2(\Omega)^4 |\  div u=0;  (w,T)=0, \frac{\de{u}}{\de z }=0\ at \ z=0,1\}.
\end{aligned}
\end{equation}

Let $L_\lambda: H_1\rightarrow H$, and $G: H_1\rightarrow H$ be defined by

\begin{equation}\label{e6}
\begin{aligned}
&L_\lambda\phi=P\left[
\begin{aligned}
&\pr(\Delta+\frac{\partial^2}{\partial{z^2}})u\\
&\pr(\Delta+\frac{\partial^2}{\partial{z^2}})w+\pr\lambda{T}\\
&(\Delta+\frac{\partial^2}{\partial{z^2}})T+\lambda{w}
\end{aligned}
\right],\\
&G(\phi)=-P\left[
\begin{aligned}
&\nabla_uu+w\frac{\partial{u}}{\partial{z}}\\
&\nabla_uw+w\frac{\partial{w}}{\partial{z}}\\
&\nabla_uT+w\frac{\partial{T}}{\partial{z}}
\end{aligned}
\right],
\end{aligned}
\end{equation}
for any $\phi=(u,w,T) \in H_1$. Here $\lambda=\sqrt{R}$, and $P:L^2(\Omega)^4\rightarrow H$ is the Leray projection.

Then the equations (\ref{e22})   with (\ref{e23}) can be rewritten in the following operator form:
\begin{equation}\label{e7}
\begin{aligned}
& \frac{\de\phi}{\de t}=L_\lambda\phi+G(\phi), \\
& \phi=\phi_0  && \text{ for } t=0.
\end{aligned}
\end{equation}
Here $\phi_0 \in H$  is the initial condition. It is then classical to prove the existence and properties of  solutions; see \cite{fmt, LTW92b}; we omit the details here.

\section{Linear Analysis and Principle of Exchange of Stability}
To study the dynamic transitions of the B\'enad convection problem (\ref{e22}) on the spherical shells, 
we need to consider the linear  eigenvalue problem:
 \begin{equation}\label{e31}L_\lambda\phi=\beta(\lambda)\phi,\end{equation}
which is equivalent to 
\begin{equation}\label{e32}
\begin{aligned}
&\pr(\Delta+\frac{\partial^2}{\partial{z^2}})u-\nabla p=\beta(\lambda)u,\\
&\pr(\Delta+\frac{\partial^2}{\partial{z^2}})w+\pr\lambda T-\fr{\de p}{\de z}=\beta(\lambda)w,\\
&(\Delta+\frac{\partial^2}{\partial{z^2}})T+\lambda w=\beta(\lambda)T,\\
&\div_3(u,w)=0,
\end{aligned}
\end{equation}
supplemented with the boundary condition (\ref{e23}).

When  $(u,w,T)$ doesn't depend on vertical $z$, then we have $(w,T)=0$, and $u=\curl Y_{l,m}$,
so the eigenvalue and eigenfunctions are as follows:
\begin{equation}\label{evef1}
\begin{aligned}
\beta_{l0}&=&-\Pr \alpha_l^2,\\
\Psi_{lm0}&=&(\curl Y_{lm},0,0).
\end{aligned}
\end{equation}
Thanks to the spherical geometry of the domain, the above eigenvalue problem can solved using  the following separation of variables:
\begin{equation}\label{e33}
\begin{aligned}
&u=\nabla{f(\theta,\varphi)}H'(z),\\
&w=\alpha^2f(\theta,\varphi)H(z),\\
&T=f(\theta,\varphi)\Theta(z).
\end{aligned}
\end{equation}
By (\ref{e33}), we infer from (\ref{e32}) and (\ref{e4})   that
\begin{align}
\label{e34}&-\Delta{f}=\alpha^2f,\\
\label{e35}&\pr(-\alpha^2\nabla{f}H'+\nabla{f}H''')-\nabla p=\beta\nabla fH',\\
\label{e36}&\pr(-\alpha^4fH+\alpha^2fH''+\lambda f\Theta)-\fr{\de p}{\de z}=\beta\alpha^2 fH,\\
\label{e37}&-\alpha^2f\Theta+f\Theta''+\lambda\alpha^2fH=\beta f\Theta,\\
\label{e38}&(\Theta,H,H'')=0\ at\ z=0,1.
\end{align}
For  the eigenvalue problem  (\ref{e34})  on the 2D sphere, the only possible candidates for the eigenfunction and eigenvalues   are given by  
\begin{align}
&\label{f}
f=Y_{lm}(\theta,\varphi),
\\
&\label{a}
\alpha_l^2=\frac{l(l+1)}{r^2}.
\end{align}
To eliminate the pressure function $p$, we do  $\fr{\de(\ref{e35})}{\de z}-\nabla(\ref{e36})$: 
\begin{equation}\label{theta}
\Theta=\fr{\pr H''''-(2\pr\alpha^2+\beta)H''+(\alpha^4\pr+\beta\alpha^2)H}{-\pr\lambda}.
\end{equation}

Then we can consider two cases. First, for the case where $l=0$, it is easy to see that the eigenvalues and eigenvectors of $L_\lambda$ are:
\begin{equation}\label{evef2}
\begin{aligned}
\beta_{0n}&=&-n^2\pi^2,\\
\Psi_{00n}&=&(0,0,\sin n\pi{z}).
\end{aligned}
\end{equation}
When $l\ne{0}$, plug (\ref{theta}) into equation (\ref{e37}), then $H, \ \Theta$ are as follows:
\begin{align}
\label{h}&H_n=\sin n\pi{z},\\
\label{H}&\Theta_n=b_{ln}\sin n\pi z.
\end{align}
Then plugging (\ref{h}) and (\ref{H}) into (\ref{theta}) and (\ref{e37}), we obtain that 
\begin{align}
\label{b}&b_{ln}=\fr{\lambda\alpha_l^2}{\beta_{ln}+\gamma_{ln}^2}=\fr{\pr \gamma_{ln}^4+\beta\gamma_{ln}^2}{-\pr\lambda},
\end{align}
where $$\gamma_{ln}^2=n^2\pi^2+\alpha_l^2=n^2\pi^2+\fr{l(l+1)}{r^2}.$$
From (\ref{b}), the eigenvalue $\beta$ solves the following quadratic equation:
$$\beta^2+(1+\pr)\gamma_{ln}^2\beta+(\gamma_{ln}^4-\fr{\lambda^2\alpha_l^2}{\gamma_{ln}^2})\pr=0$$
 we get two types of eigenvalues and eigenvectors:
 \begin{equation}\label{evef3}
\begin{aligned}
&\beta^{+}_{ln}=\frac{1}{2}[-\gamma_{ln}^2(1+\pr)+(\gamma_{ln}^4(1+\pr)^2-4\pr(\gamma_{ln}^4-\frac{\lambda^2\alpha_l^2}{\gamma_{ln}^2}))^{1/2}],\\
&\Psi_{lmn}^{+}=(n\pi{\cos n\pi{z}}\nabla{Y}_{lm},\  \alpha_l^2Y_{lm}\sin n\pi{z},\  \fr{\lambda\alpha_l^2}{\beta_{ln}^++\gamma_{ln}^2}Y_{lm}\sin n\pi{z}),\\
&\beta_{ln}^{-}=\frac{1}{2}[-\gamma_{ln}^2(1+\pr)-(\gamma_{ln}^4(1+\pr)^2-4\pr(\gamma_{ln}^4-\frac{\lambda^2\alpha_l^2}{\gamma_{ln}^2}))^{1/2}],\\
&\Psi_{lmn}^{-}=(n\pi{\cos n\pi{z}}\nabla{Y}_{lm},\  \alpha_l^2Y_{lm}\sin n\pi{z},\  \fr{\lambda\alpha_l^2}{\beta_{ln}^-+\gamma_{ln}^2}Y_{lm}\sin n\pi{z}).
\end{aligned}
\end{equation}

So for problem (\ref{e22}) with boundary condition (\ref{e23}), the eigenfunctions for the corresponding linear operator as in (\ref{e32}) are consist of three group, which is shown in (\ref{evef1}), (\ref{evef2}), and (\ref{evef3}).

Now we define the critical Rayleigh number $R_c=\lambda_c^2$  by
 \begin{equation}\label{ld}
\lambda_c=\min_{l,n}\frac{(n^2\pi^2+\alpha_l^2)^{3/2}}{\alpha_l}=\min_l\fr{(\pi^2+\alpha_l^2)^{3/2}}{\alpha_l}.
\end{equation}
It is easy to see that  the minimum in the above formula is  $\lambda_c$ is achieved at $\alpha_{l_c}^2$. By definition, we have (see (\ref{wavenumber})):
\begin{equation} \la{cwn}
\left[\frac{h}{a}\right]^2 l_c(l_c+1)=\frac{\pi^2}{2}.
\end{equation}
As $l$ can only take integer values, it is possible for the minimum in (\ref{ld}) is achieved at two consecutive $l$'s. For simplicity and without loss of generality, 
in this article,  we only consider the case where  the minimum is achieved at $\alpha_{l_c}$ with an 
 integer $l_c$:
\begin{equation}
\label{lc}
\lambda_c=\fr{(\pi^2+\alpha_{l_c}^2)^{3/2}}{\alpha_{l_c}} =\frac{3\sqrt{3}\pi^2}{2} \qquad \text{ with } \qquad \alpha_{l_c}= \frac{l_c(l_c+1)}{r^2},
\end{equation}
and the critical Rayleigh number is 
\begin{equation}
R_c=\lambda_c^2 = \frac{27\pi^4}{4}=657.5.\la{rc}
\end{equation}
This value is consistent with the value developed classically; see among many others \cite{chandrasekhar, dr}.

By definition, the following principle of exchange of stability (PES) holds true:
\begin{equation*}
\begin{aligned}
&\beta_{l_c1}^+(\lambda)\left\{\begin{aligned}
&<0&&\text{if} \ \lambda<\lambda_c,\\
&=0&&\text{if}\ \lambda<\lambda_c,\\
&>0&&\text{if} \  \lambda>\lambda_c,
\end{aligned}\right.\\
&\beta_{ln}^+(\lambda_c)<0&& \forall (l,n)\ne (l_c,1),\ l,n>0\\
&\beta_{ln}^-(\lambda_c)<0 &&\forall \  (l,n),\ l,n>0,\\
&\beta_{l0}(\lambda_c)<0 &&\forall\  l> 0,\\
&\beta_{0n}(\lambda_c)<0 &&\forall\  n> 0.
\end{aligned}
\end{equation*}

The critical eigenspace is then given by 
\begin{equation}\la{space-e1}
E_1=\{\sum_{m=-l_c}^{l_c}x_m\Psi_{l_cm1}^+\ |\ x_{-m}=(-1)^mx_m^*,\ x_m\in\mathds{C},\ |m|\le l_c\} 
\end{equation}
with dimension dim$E_1=2l_c+1$.

\section{Dynamical Transitions}
As mentioned in the introduction, the linear operator $L_\lambda$ is symmetric, leading to the real eigenvalues $\beta_{l,n}^\pm$, $\beta_{0,n}$. With the PES at our disposal,  we can obtain immediately that the B\'enard convection system will undergoes a dynamic transition, and the following theorem shows that the transition is always Type-I. 
\begin{thm}\la{t4.1}
The B\'enard convection problem (\ref{e7}) always undergoes a Type-I dynamic transition as the Rayleigh number $R$ crosses the critical value ${R}_c$. In particular, the following assertions hold true:

\begin{enumerate}

\item When $R\le {R}_c$, the basic motionless state with linear temperature profile $ \phi=(u,w,T)=0$ is globally asymptotically stable.

\item As $R>R_c$, but near $R_c$, the problem bifurcates to an attractor $\Sigma_R$, which is 
an $2l_c$-dimensional homological sphere, where the integer $l_c$  is defined  by (\ref{lc}).
 
\item For any $\phi=(u,w,T)\in\Sigma_R$,
$$\phi=\sum_{m=-1_c}^{l_c} {x_m}\Psi_{l_cm1}^{+}   +o(\sum_{m=-l_c}^{l_c} |x_m|),\,\,\,\, x_{-m}=(-1)^mx_m^*,\,\,\, x_m\in \mathds{C},\ |m|\le l_c$$

\item The attractor $\Sigma_R$ attracts $H\setminus\Gamma$, where $\Gamma$ is the stable manifold of $\phi=0$ with codimension $2l_c+1$.

\item The attractor $\Sigma_R$ contains at least one steady state solution of the problem.
\end{enumerate}

\end{thm}

\begin{proof}
As in \cite{MW04d},  we can show that the basic motionless state is (globally) asymptotically stable  the critical eigenvalue  $\lambda=\lambda_c$. By the attractor bifurcation theorem in \cite{b-book}  and the linear analysis in the previous subsection, we obtain immediately Assertions (1)--(4) in the theorem.

It suffices then to  show the existence of at least one steady state in the bifurcated attractor $\Sigma_R$. 
To this end,  the steady state problem for (\ref{e7}) can be written as 
\begin{equation}\label{ss}
-A\phi - B_\lambda\phi =G(\phi),
\end{equation}
where $A + B_\lambda =L_\lambda$ and the linear operator $A$  and $B_\lambda$  are defined by 
\begin{equation}
A\phi=P\left[
\begin{aligned}
&\pr(\Delta+\frac{\partial^2}{\partial{z^2}})u \\
&\pr(\Delta+\frac{\partial^2}{\partial{z^2}})w\\
&(\Delta+\frac{\partial^2}{\partial{z^2}})T
\end{aligned}
\right],
\qquad 
B_\lambda\phi=P\left[
\begin{aligned}
&0 \\
&\pr\lambda{T}\\
&\lambda{w}
\end{aligned}
\right],\qquad \forall \phi=(u,w,T) \in H_1.
\end{equation}
Obviously, the operator  $A$  is invertible, and (\ref{ss})  can be written as 
\begin{equation} \la{ss1}
[\text{id} + (-A)^{-1} B_\lambda] \phi = (-A)^{-1}  G(\phi). 
\end{equation}
Then it is classical to show that the linear operator $\text{id} + (-A)^{-1} B_\lambda: H \to H$  is a completely continuous field, and $(-A)^{-1}  G: H \to H$ is a compact operator. 
Hence by  the Krasnoselski bifurcation theorem (see e.g. Theorem 1.10 in \cite{b-book}),  
(\ref{ss1}) has a nontrivial steady state bifurcation at $R_c$, as the critical eigenvalue $R_c$ has odd multiplicity $2l_c+1$, which is the dimension of the eigenspace $E_1$.  The proof is complete.\end{proof}

\section{Structure and Patterns of the B\'enard Convection}
The above theorem shows that the system always undergoes a dynamic transition as $R$ crosses  the critical Rayleigh number $R_c$,  and the transition states occupy a set $\Sigma_R$, homological to $S^{2l_c}$. 

The main objective of this section is to derive more detailed structure of this bifurcated attractor. In particular, we shall give a strategy to show that $\Sigma_R$ is in fact {\it homeomorphic} to $S^{2l_c}$. 

The approach for this purpose is reduce the governing partial differential equations to the center manifold generated by the unstable modes for $R> R_c$ and near $R_c$. As the general case is tedious and time consuming, we present here only two special cases  where critical wave number $l_c=1$  or $2$. The other cases can be studied in the same fashion. 

\begin{thm} \la{t2.2}
Consider the B\'enard convection in the case where $l_c=1$, which is equivalent in the dimensional form to 
$$\frac{h}{a}=\frac{\pi}{2},$$
where $a$  is the horizontal length scale and $h$ is the vertical scale. Then the bifurcated attractor $\Sigma_R$ is homeomorphic  to $S^2$, consisting of only degenerate steady states. Furthermore, $\Sigma_R$  is approximated by
$$
\Sigma_R\simeq \{ \sum_{m=-1}^1\Psi_{1m1}^+  \quad | \quad x_{-m}=(-1)^mx_m^*,\ x_0|^2+|x_1|^2+|x_{-1}|^2=\beta_{11}^+/q_1 \},
$$
where $q_1$ is given by
$$
q_1=\fr{3 \pi^3 (17787 + 355912 \pr - 669713 \pr^2 + 387787 \pr^3)}{400000 \pr (353 - 625 \pr + 353 \pr^2)}.
$$
Namely, for any $\phi\in \Sigma_R$,
\begin{align*}
&\phi=\sum_{m=-1}^1x_m\Psi_{1m1}^++o((\beta_{11}^+/q_1)^{1/2}),\\
&x_0^2+|x_{-1}|^2+|x_1|^2=\beta_{11}^+/q_1,\\
&x_{-m}=(-1)^mx_m^*,\,\,\,  |m|\le1.
\end{align*}
\end{thm}

\begin{thm} \la{t2.3}
For the case where $l_c=2$, equivalently 
$$\frac{h}{a}=\frac{\pi}{2\sqrt{3}},$$
the bifurcated attractor $\Sigma_R$ is homeomorphic  to $S^4$, which consists of at least an $S^2$ subset of  degenerate steady states. Furthermore, when $\Pr=1$, $\Sigma_R$ is approximated by 
$$
\Sigma_R\simeq \{\sum_{m=-2}^2\Psi_{2m2}^+ \quad | \quad x_{-m}=(-1)^mx_m^*,\ \sum_{m=-2}^2|x_m|^2=\beta_{21}^+/q_2 \},
$$
where $q_2$ is given by
$$
q_2=\fr{2291405 \pi^3 }{214754176}.
$$
\end{thm}

The proof of these two theorems is based on the reduction of the original Boussinesq equations to the center manifold generated by  the first unstable modes, and to study the dynamics of the reduced system. As mentioned, the same method applies to other cases as well.  

\bp[Proof of Theorem \ref{t2.2}] 
We proceed in several steps as follows.

{\sc Step 1. General Center Manifold Reduction Stratege.}
More precisely, let
\begin{align*}
& E_1=\left\{ x=\sum_{m=-l_c}^{l_c}x_m\Psi_{l_cm1}^+\  |\  x_{-m}=(-1)^mx_m^*,\ x_m\in\mathds{C},\ |m|\le l_c\right\},\\
& E_2=E^\perp_1.
%\{\sum_{l\ge 1,\ |m|\le l}y_{lm0}\Psi_{lm0}+\sum_{n\ge1}y_{00n}\Psi_{00n}+\sum_{(l,n)\ne {l_c,1},l,n>1}\sum_{|m|\le l}y_{lmn}^+\Psi_{lmn}^+\\
%&+\sum_{l,n>1}\sum_{|m|\le l}y_{lmn}^-\Psi_{lmn}^-,\  y_{lm0},y_{00n},y_{lmn}^+,y_{lmn}^-\in\mathds{C}\}.
\end{align*}

The main idea of the center manifold function is to seek a function $\Phi$, which maps a neighborhood  of the original in $E_1$ into $E_2$ given by 
\begin{align}\label{CMF}
\Phi(x,\lambda)=  & \sum_{l\ge 1,\ |m|\le l}y_{lm0}\Psi_{lm0}+\sum_{n\ge1}y_{00n}\Psi_{00n}\\
&  +\sum_{(l,n)\ne {l_c,1},l,n>1}\sum_{|m|\le l}y_{lmn}^+\Psi_{lmn}^+
+\sum_{l,n>1}\sum_{|m|\le l}y_{lmn}^-\Psi_{lmn}^-. \nonumber
\end{align}
By the approximation formula derived in \cite{ptd}, the center manifold function  $\Phi$ can be approximated by 
\begin{equation}\la{cmf-app}
-L_\lambda\Phi(x,\lambda)=P_2G(x)+o(2),
\end{equation}
where $P_2:H\to E_2$  is the canonical projection, and 
$$o(2)=o(||x||^2)+O(|\beta_{l_c1}^+(\lambda) |\cdot ||x||^2).$$
The first term in the right-hand side of (\ref{cmf-app}) involves products of the eigenfunctions $\Psi_{l_cm1}^+$, which are defined in terms of the spherical harmonics $Y_{lm}$  with $l \le l_c$. Hence by the properties of the 3-j symbols for the products of spherical harmonics, we only have to calculate the following center manifold function coefficients and all the others are zero or of higher-order:
\begin{enumerate}
\item For $1 \le l \le 2l_c,\ |m|\le l$,
\begin{equation}\label{yl0}
y_{lm0}=-\fr{<G(x), \Psi_{lm0}>}{\beta_{l0}||\Psi_{lm0}||^2}+o(2);
\end{equation}
\item  For $n=2$,
\begin{equation}\label{y02}
y_{002}=-\fr{<G(x), \Psi_{002}>}{\beta_{02}||\Psi_{002}||^2}+o(2);
\end{equation}
\item For $l\ne l_c$, $1\le l \le 2l_c$, $|m|\le l$, and $n=0\ or\ 2,$ 
\begin{equation}\label{y+}
y_{lmn}^+=-\fr{<G(x), \Psi_{lmn}^+>}{\beta_{ln}^+||\Psi_{lmn}^+||^2}+o(2);
\end{equation}
\item For  $1\le l \le 2l_c$, $|m|\le l$, and $n=0 \ or \ 2$,
\begin{equation}\label{y-}
y_{lmn}^-=-\fr{<G(x), \Psi_{lmn}^->}{\beta_{ln}^-||\Psi_{lmn}^-||^2}+o(2).
\end{equation}

\end{enumerate}
We note that hand-calculation of the right had sides of above formulas are still difficult if not impossible. Fortunately, one can write a simple mathematica program to perform symbolic calculations to derive precise formulas these coefficients.

Finally, the  reduced equations are given as follows:
\begin{equation}\label{reduced eq} 
\fr{dx_m}{dt}=\beta_{l_c1}^+(\lambda)x_m+\fr{<G(X+\Phi(x,\lambda)),\Psi_{l_cm1}^+>}{<\Psi_{l_cm1}^+,\Psi_{l_cm1}^+>} \qquad  \text{for } \ |m|\le l_c.
\end{equation}
Again, to carry out the calculation for the nonlinear interaction in the right-hand side of this reduced equation, we need to use mathematica as well.

\medskip

{\sc Step 2. Approximation of the center manifold function for $l_c=1$.}
In this case,  the first eigenvalue of (\ref{e31}) with (\ref{e33}) is 
$$\beta_{11}^+(\lambda)=\fr{1}{12}[-9 \pi^2 (1 + \pr) + \sqrt{81 \pi^4 (-1 + \pr)^2 + 48 \pr \lambda^2}].$$
Using Mathematica, we infer  from (\ref{yl0})-(\ref{y-}) that:
\begin{equation}\la{cmf-coe}
\begin{aligned}
&y_{2-22}^+(x,\lambda)= \fr{1}{-\beta_{22}^+} d_1^+ x_{-1}^2, 
     && y_{2-22}^-(x,\lambda)=\fr{1}{-\beta_{22}^-} d_1^-x_{-1}^2,\\
&y_{2-12}^+(x,\lambda)= \fr{1}{-\beta_{22}^+}\sqrt{2}d_1^+x_{-1}x_0, 
     &&  y_{2-12}^-(x,\lambda)= \fr{1}{-\beta_{22}^-}\sqrt{2}d_1^-x_{-1}x_0,\\
&y_{202}^+(x,\lambda)=\fr{1}{-\beta_{22}^+}\sqrt{\fr{3}{2}}d_1^+(x_0^2+x_{-1}x_1), 
     &&  y_{202}^-(x,\lambda)=\fr{1}{-\beta_{22}^-}\sqrt{\fr{3}{2}}d_1^-(x_0^2+x_{-1}x_1),\\
&y_{212}^+(x,\lambda)=\fr{1}{-\beta_{22}^+}\sqrt{2}d_1^+x_0x_1, 
    &&  y_{212}^-(x,\lambda)=\fr{1}{-\beta_{22}^-}\sqrt{2}d_1^-x_0x_1,\\
&y_{222}^+(x,\lambda)=\fr{1}{-\beta_{22}^+}d_1^+ x_1^2, 
    && y_{222}^-(x,\lambda)=\fr{1}{-\beta_{22}^-}d_1^-x_1^2,\\
&y_{002}=-\fr{1}{64}\sqrt{3}\pi^2(x_0^2-2x_{-1}x_1),\\
&y_{lmn}^\pm=0,  &&  \text{when} \  (l,n)\ne(2,2), \ l>0.
\end{aligned}
\end{equation}
where
\begin{align*}
& d_1^+=3 \sqrt{\fr{3}{10}}\fr{\pi^{5/
  2} (121 - 121 \pr + 
   A) (187 - 121 \pr + 
   A)}{
1936 [-1493 - 1331 \pr^2 - 
   11 A + 
   \pr (2500 + 11 A)]},\\
&   d_1^-= -3\sqrt{\fr{3}{10}} (\pi^{5/
  2}) \fr{(-121 + 121 \pr + 
  A) (-187 +121 \pr + 
   A)}
{1936 [1493 +1331 \pr^2 - 
   11 A + 
   \pr (-2500 + 11 A)]},\\
& \beta_{2,2}^+= \fr{1}{44} \pi^2 (-121 - 121 \pr + 
  A), \\
&    \beta_{2,2}^-=-\fr{1}{44} \pi^2 (121 + 121 \pr + 
  A), \\
  &  A=\sqrt{11}\sqrt {1331 - 2338\pr + 1331 \pr^2}.
\end {align*}
Hence we obtain immediately the following second-order approximation of the center manifold function for  
$\lambda\ge \lambda_c$: 
\begin{align}\label{CMF1}
\Phi(x,\lambda)=&y _{2-22}^+\Psi_{2-22}^++y_{2-12}^+\Psi_{2-12}^++y_{202}^+\Psi_{202}^+\\
&+y_{212}^+\Psi_{212}^++y_{222}^+\Psi_{222}^+   \nonumber \\
&+y_{2-22}^-\Psi_{2-22}^-+y_{2-12}^-\Psi_{2-12}^-+y_{202}^-\Psi_{202}^- \nonumber \\
&+y_{212}^-\Psi_{212}^-+y_{222}^-\Psi_{222}^-+y_{002}\Psi_{002}+o(2), \nonumber
\end{align}
with the coefficients given by (\ref{cmf-coe}). 

\medskip

{\sc Step 3. Reduced equations.} With the approximation of the center manifold function above, we deduce immediately  following reduced equations:
\begin{equation}
\begin{aligned}
& \fr{dx_{-1}}{dt}=\beta_{11}^+x_{-1}- q_1x_{-1}(x_0^2-2x_{-1}x_1)+o(3)\\
& \fr{dx_0}{dt}=\beta_{11}^+x_0- q_1x_0(x_0^2-2x_{-1}x_1)+o(3),\\
& \fr{dx_1}{dt}= \beta_{11}^+x_1- q_1x_1(x_0^2-2x_{-1}x_1)+o(3),
\end{aligned}
\end{equation}
where 
\begin{align*}
& q_1=\fr{3 \pi^3 (17787 + 355912 \pr - 669713 \pr^2 + 387787 \pr^3)}{400000 \pr (353 - 625 \pr + 353 \pr^2)},\\
& o(3)=o(||x||^3)+O(|\beta_{l_c1}^+| \cdot ||x||^3).
\end{align*}
As the eigenfunction are complex, we can set 
$$x_m=\fr{1}{\sqrt{2}}(y_m+iz_m),\qquad y_m, \ z_m\in\mathbb{R}.$$
By definition (see (\ref{space-e1})), we  have
$$y_m=(-1)^my_{-m}, \ z_m=(-1)^{m+1}z_{-m},\ m>0.$$ 
Hence we obtain the following   reduced equations:
\begin{equation}\la{reduced}
\begin{aligned}
\fr{dy_1}{dt}=&\beta_{11} ^+y_1 - q_1y_1(x_0^2+y_1^2+z_1^2)+o(3),\\
\fr{dx_0}{dt}=&\beta_{11} ^+x_0 - q_1x_0(x_0^2+y_1^2+z_1^2)+o(3),\\
\fr{dz_1}{dt}=&\beta_{11} ^+z_1 - q_1z_1(x_0^2+y_1^2+z_1^2)+o(3).
\end{aligned}
\end{equation}

\medskip

{\sc Step 4.} Let $g$  be the cubic terms in the right hand side of the reduced equation. Then, we have  
$$ - (g, (y_1, x_0, z_1)) \le q_1 (x_0^2+y_1^2+z_1^2)^2.$$
Hence by the attractor bifurcation theorem in \cite{ptd},  the bifurcated attractor  $\Sigma_R$ is homeomorphic to $S^2$. 

Also, by Theorem~\ref{t4.1}, there is at least one steady state solution
in $\Sigma_R$. In addition, due to the spherical geometry, the original Boussinesq equations have an $S^2$-symmetry. Hence, this steady state solution generates an $S^2$-set of steady states. Therefore,  
$\Sigma_R$ consists precisely steady state solutions. 

Finally, by setting the right hand-side of the above reduced equations (\ref{reduced}) equal to zero, we obtain that $\Sigma_R$  is approximated by 
$$\{ (y_1, x_0, z_1) \in \mathbb R^3 \quad | \quad x_0^2+y_1^2+z_1^2= \beta_{11}^+/q_1 \}.
$$

The proof is complete.
\ep

\bp[Proof of Theorem~\ref{t2.3}] 
We note that the unstable space $E_1$  is now  given by 
$$ E_1=\left\{ x=\sum_{m=-2}^{2}x_m\Psi_{2m1}^+\  |\  x_{-m}=(-1)^mx_m^*,\ x_m\in\mathds{C},\ |m|\le l_c\right\}, $$
and we still use for brevity $x$ to denote 
$$x=\sum_{m=-2}^{2}x_m\Psi_{2m1}^+.$$

 As in the previous proof, the approximation of the center manifold function  is then given by 
\begin{align}\label{CMF2}
\Phi(x,\lambda)=&y_{2-22}^+\Psi_{2-22}^++y_{2-12}^+\Psi_{2-12}^++y_{202}^+\Psi_{202}^++y_{212}^+\Psi_{212}^++y_{222}^+\Psi_{222}^+\\
&+y_{4-42}^+\Psi_{4-42}^++y_{4-32}^+\Psi_{4-32}^++y_{4,-2,2}^+\Psi_{4-22}^++y_{4-12}^+\Psi_{4-12}^++y_{402}^+\Psi_{402}^+ \nonumber \\
&+y_{412}^+\Psi_{412}^++y_{422}^+\Psi_{422}^++y_{4,3,2}^+\Psi_{432}^++y_{442}^+\Psi_{442}^+ \nonumber \\
&+y_{2-22}^-\Psi_{2-22}^-+y_{2-12}^-\Psi_{2-12}^-+y_{2,0,2}^-\Psi_{202}^-+y_{212}^-\Psi_{212}^-+y_{222}^-\Psi_{222}^-  \nonumber \\
&+y_{4-42}^-\Psi_{4-42}^-+y_{4-32}^-\Psi_{4-32}^-+y_{4-22}^-\Psi_{4-22}^-+y_{4-12}^-\Psi_{4-12}^-+y_{402}^-\Psi_{402}^- \nonumber \\
&+y_{412}^-\Psi_{412}^-+y_{422}^-\Psi_{422}^-+y_{432}^-\Psi_{432}^-+y_{442}^-\Psi_{442}^-+y_{002}\Psi_{002}+o(2), \nonumber 
\end{align}
where the  non vanishing coefficients  in the above center manifold function are as follows:
\begin{align*}
&y_{2-22}^+(x,\lambda_c)= \fr{ c_1^+(\sqrt{6} x_{-1}^2 - 4 x_{-2} x_0)}{-\beta_{22}^+},&&y_{2-22}^-(x,\lambda_c)= \fr{ c_1^-(\sqrt{6} x_{-1}^2 - 4 x_{-2} x_0}{-\beta_{22}^-},\\
&y_{2-12}^+(x,\lambda_c)= \fr{ 2c_1^+(x_{-1} x_0 - \sqrt{6} x_{-2} x_1)}{-\beta_{22}^+},&&y_{2-12}^-(x,\lambda_c)= \fr{ 2c_1^-(x_{-1} x_0 - \sqrt{6} x_{-2} x_1)}{-\beta_{22}^-},\\
&y_{202}^+(x,\lambda_c)= \fr{2c_1^+(x_0^2 - x_{-1} x_1 - 2 x_{-2} x_2)}{-\beta_{22}^+},&& y_{202}^-(x,\lambda_c)= \fr{2c_1^-(x_0^2 - x_{-1} x_1 - 2 x_{-2} x_2)}{-\beta_{22}^-},\\
&y_{212}^+(x,\lambda_c)=\fr{2c_1^+(x_0 x_1 - \sqrt{6} x_{-1} x_2)}{-\beta_{22}^+},&&y_{212}^-(x,\lambda_c)=\fr{2c_1^-(x_0 x_1 - \sqrt{6} x_{-1} x_2)}{-\beta_{22}^-},\\
&y_{222}^+(x,\lambda_c)=\fr{c_1^+(\sqrt{6} x_1^2 - 4 x_0 x_2)}{-\beta_{22}^+},&&y_{222}^-(x,\lambda_c)=\fr{c_1^-(\sqrt{6} x_1^2 - 4 x_0 x_2)}{-\beta_{22}^-},\\
&y_{4-42}^+(x,\lambda_c)=\fr{c_2^+}{-\beta_{42}^+}x_{-2}^2,&&y_{4-42}^-(x,\lambda_c)=\fr{c_2^-}{-\beta_{42}^-}x_{-2}^2,\\
&y_{4-32}^+(x,\lambda_c)=\fr{\sqrt{2}c_2^+}{-\beta_{42}^+}x_{-2}x_{-1},&& y_{4-32}^-(x,\lambda_c)=\fr{\sqrt{2}c_2^-}{-\beta_{42}^-}x_{-2}x_{-1},\\
&y_{4-22}^+(x,\lambda_c)=\fr{\sqrt{5}c_3^+}{-\beta_{42}^+}(\sqrt{2} x_{-1}^2 + \sqrt{3} x_{-2} x_0),&&y_{4-22}^-(x,\lambda_c)=\fr{\sqrt{5}c_3^-}{-\beta_{42}^-}(\sqrt{2} x_{-1}^2 + \sqrt{3} x_{-2} x_0),\\
&y_{4-12}^+(x,\lambda_c)=\fr{\sqrt{5}c_3^+}{-\beta_{42}^+}(\sqrt{6} x_{-1} x_0 + x_{-2} x_1),&&y_{4-12}^-(x,\lambda_c)=\fr{\sqrt{5}c_3^-}{-\beta_{42}^-}(\sqrt{6} x_{-1} x_0 + x_{-2} x_1),\\
&y_{402}^+(x,\lambda_c)=\fr{c_3^+}{-\beta_{42}^+}(3 x_0^2 + 4 x_{-1} x_1 + x_{-2} x_2),&&y_{402}^-(x,\lambda_c)=\fr{c_3^-}{-\beta_{42}^-}(3 x_0^2 + 4 x_{-1} x_1 + x_{-2} x_2),\\
&y_{412}^+(x,\lambda_c)=\fr{\sqrt{5}c_3^+}{-\beta_{42}^+}(\sqrt{6} x_0x_1 + x_{-1} x_2),&&y_{412}^-(x,\lambda_c)=\fr{\sqrt{5}c_3^-}{-\beta_{42}^-}(\sqrt{6} x_0x_1 + x_{-1} x_2),\\
&y_{422}^+(x,\lambda_c)=\fr{\sqrt{5}c_3^+}{-\beta_{42}^+}(\sqrt{2} x_1^2 + \sqrt{3} x_0 x_2),&& y_{422}^-(x,\lambda_c)=\fr{\sqrt{5}c_3^-}{-\beta_{42}^-}(\sqrt{2} x_1^2 + \sqrt{3} x_0 x_2),\\
&y_{432}^+(x,\lambda_c)=\fr{\sqrt{2}c_2^+}{-\beta_{42}^+}x_{1}x_2,&& y_{432}^-(x,\lambda_c)=\fr{\sqrt{2}c_2^-}{-\beta_{42}^-}x_1x_2,\\
&y_{442}^+(x,\lambda_c)=\fr{c_2^+}{-\beta_{42}^+}x_2^2,&&y_{442}^-(x,\lambda_c)=\fr{c_2^-}{-\beta_{42}^-}x_2^2,\\
&y_{002}=-\fr{1}{16} \sqrt{3}\pi^4 (x_0^2 - 2 x_{-1} x_1 + 2 x_{-2} x_2),
%&y_{lmn}^+=0, \  y_{lmn}^-=0, \ when\  (l,n) \ne(2,2) \ or \ (4,2),\ l>0.
\end{align*}
where
\begin{equation}
\begin{aligned}
c_1^+=&\fr{\sqrt{5} \pi^{
 5/2}(9 - 9 \pr + B) (15 - 9 \pr + B)}{672 [-29 - 27 \pr^2 - 3 B + 
  \pr (52 + 3B)]},\\
 c_2^+=&\fr{3 \sqrt{\fr{5}{14}} \pi^{
 5/2} (-1 - \fr{153}{(
   289 - 289 \pr + \sqrt{17}C)})}{136 (1 + \fr{20655}{(289 - 289 \pr + 
     \sqrt{17}C)^2})},\\
  c_3^+=&\fr{3 \pi^{5/
  2} (442 - 289 \pr + 
   \sqrt{17} C) (289 - 289 \pr + 
   \sqrt{17} C)}{16184 (-11041 - 9826 \pr^2 - 
   34 \sqrt{17} C + 
   \pr (18437 + 34 \sqrt{17} C))},\\
   c_1^-=&-\fr{\sqrt{5} \pi^{
 5/2}(-9 + 9 \pr +B) (-15 + 9 \pr +B)}{672 (29 + 27 \pr^2 - 3 B + 
  \pr (-52 + 3B))},\\
 c_2^-=&\fr{3 \sqrt{\fr{5}{14}} \pi^{
 5/2} (-1 + \fr{153}{(
   -289 + 289 \pr + \sqrt{17} C)})}{136 (1 + \fr{20655}{(-289 + 289 \pr + 
     \sqrt{17}C)^2})},\\
  c_3^-=&\fr{-3 \pi^{5/
  2} (-442 + 289 \pr + 
   \sqrt{17} C) (-289 + 289 \pr + 
   \sqrt{17} C)}{16184 (11041 + 9826 \pr^2 - 
   34 \sqrt{17} C + 
   \pr (-18437 + 34 \sqrt{17} C))}.
\end{aligned}
\end{equation}
Here
\begin{align*}
&B= \sqrt{
   81 - 150 \pr + 81 \pr^2},\\
 &C= \sqrt{4913 - 8611 \pr + 4913 \pr^2},
\end{align*}
and $\beta_{2,2}^+, \beta_{4,2}^+, \beta_{2,2}^-, \beta_{4,2}^-$ are given by (\ref{evef3}). 

For simplicity, we only consider the case where $\Pr=1$. 
Then the reduced equation are as following:
\begin{equation}
\begin{aligned}
\fr{dx_{-2}}{dt}=&\beta_{21}^+x_{-2}-q_2x_{-2} (x_0^2 - 2 x_{-1} x_1 + 2 x_{-2} x_2)+o(3),\\
\fr{dx_{-1}}{dt}=&\beta_{21}^+x_{-1}-q_2x_{-1} (x_0^2 - 2 x_{-1} x_1 + 2 x_{-2} x_2)+o(3),\\
\fr{dx_0}{dt}=&\beta_{21}^+x_{0}-q_2x_{0} (x_0^2 - 2 x_{-1} x_1 + 2 x_{-2} x_2)+o(3),\\
\fr{dx_1}{dt}=&\beta_{21}^+x_{1}-q_2x_{1} (x_0^2 - 2 x_{-1} x_1 + 2 x_{-2} x_2)+o(3),\\
\fr{dx_2}{dt}=&\beta_{21}^+x_{2}-q_2x_{2} (x_0^2 - 2 x_{-1} x_1 + 2 x_{-2} x_2)+o(3),
\end{aligned}
\end{equation}
where 
$$q_2=\fr{2291405 \pi^3 }{214754176}.$$
Let 
$$x_m=\fr{1}{\sqrt{2}}(y_m+iz_m),\qquad y_m, \ z_m\in\mathbb{R}.$$
Then 
$$y_m=(-1)^my_{-m}, \ z_m=(-1)^{m+1}z_{-m},\ m>0.$$ 
We derive then the reduced equations  as follows:
\begin{equation}\la{reduced2}
\begin{aligned}
\fr{dy_1}{dt}=&\beta_{21}^+y_1-q_2y_1(x_0^2+y_1^2+z_1^2+y_2^2+z_2^2)+o(3),\\
\fr{dy_2}{dt}=&\beta_{21}^+y_2-q_2y_2(x_0^2+y_1^2+z_1^2+y_2^2+z_2^2)+o(3),\\
\fr{dx_0}{dt}=&\beta_{21}^+x_0-q_2x_0(x_0^2+y_1^2+z_1^2+y_2^2+z_2^2)+o(3),\\
\fr{dz_1}{dt}=&\beta_{21}^+z_1-q_2z_1(x_0^2+y_1^2+z_1^2+y_2^2+z_2^2)+o(3),\\
\fr{dz_2}{dt}=&\beta_{21}^+z_2-q_2z_2(x_0^2+y_1^2+z_1^2+y_2^2+z_2^2)+o(3).\\
\end{aligned}
\end{equation}
Then as in the proof of last theorem, we can show that $\Sigma_R$  is homeomorphic to $S^4$, and we have  the following approximation:
$$\Sigma_R\simeq \{(y_2, y_1, x_0, z_1, z_2) \in \mathbb R^5 \quad | \quad 
x_0^2+y_1^2+z_1^2+y_2^2+z_2^2 =\beta_{21}^+/q_2 \},
$$
and $\Sigma_R$ contains at least an $S^2$ subset of steady state solutions. We remark here that if we retain only the linear and cubic order terms in the above reduced equations (\ref{reduced2}), 
we would have the whole $S^4$ consisting only steady state solutions. However, since the steady states solutions are degenerate (with zero Jacobian), there is no persistence and stability of these steady state solutions.

The proof of the theorem is complete.

\ep

\section{Large-Scale Circulations and Turbulent Frictions}
Now we consider the following equations:
 \begin{equation}\label{e3}
 \begin{aligned}
 &\frac{1}{\pr}(u_t+\nabla_uu+w\frac{\partial{u}}{\partial{z}}+\nabla{p})+\sigma_0u-(\Delta+\frac{\partial^2}{\partial{z}^2})u=0,\\
 &\frac{1}{\pr}(w_t+\nabla_uw+w\frac{\partial{w}}{\partial{z}}+\frac{\partial{p}}{\partial{z}})+\sigma_1w-(\Delta+\frac{\partial^2}{\partial{z}^2})w-\sqrt{R}T=0,\\
 &T_t+\nabla_uT+w\frac{\partial{w}}{\partial{z}}-\sqrt{R}w-(\Delta+\frac{\partial^2}{\partial{z}^2})T=0,\\
 &\div{u}+\frac{\partial{w}}{\partial{z}}=0,
 \end{aligned}
 \end{equation}

supplemented with boundary condition:
 \begin{equation}\label{e4}w=0, T=0,  \frac{\partial{u}}{\partial{z}}=0\ at \ z=0,1.\end{equation}
By similar computation, we get the three groups of eigenvalues and eigenvectors:
\begin{enumerate}
\item For (w,T)=0
\begin{align*}
&\beta_{l0}=-\Pr(\alpha_l^2+\sigma_0),\\
&\Psi_{im0}=(\curl Y_{lm},0,0);
\end{align*}
\item For $l=0$ with $n\ge1$
\begin{align*}
&\beta_{0n}=-n^2\pi^2,\\
&\Psi_{00n}=(0, 0, \sin n\pi z);
\end{align*}
\item For $l,n\ge 1$
\begin{align*}
&\beta^{+}_{ln}=-\frac{1}{2}D
+\frac{1}{2}\sqrt{D^2-4\pr E},\\
&\Psi_{lmn}^{+}=(n\pi{\cos n\pi{z}}\nabla{Y}_{lm},\  \alpha_l^2Y_{lm}\sin n\pi{z},\  \fr{2\lambda\alpha_l^2}{\beta_{ln}^++\gamma_{ln}^2}Y_{lm}\sin n\pi{z}),\\
&\beta_{ln}^{-}=-\frac{1}{2}D
-\frac{1}{2}\sqrt{D^2-4\pr E},\\
&\Psi_{lmn}^{-}=(n\pi{\cos n\pi{z}}\nabla{Y}_{lm},\  \alpha_l^2Y_{lm}\sin n\pi{z},\  \fr{2\lambda\alpha_l^2}{\beta_{ln}^-+\gamma_{ln}^2}Y_{l,m}\sin n\pi{z}).
 \end{align*}
\end{enumerate}
Here
\begin{align*}
&D=\gamma_{ln}^2(1+\pr)+\fr{\pr(\sigma_1\alpha_l^2+\sigma_0n^2\pi^2)}{\gamma_{ln}^2},\\
&E=(\gamma_{ln}^4-\frac{\lambda^2\alpha_l^2}{\gamma_{ln}^2}+\sigma_1\alpha_l^2+\sigma_0n^2\pi^2).
\end{align*}
 To get the critical Rayleigh number $\lambda_c$, let $\beta_{l,n}=0,$
and obtain that 
\begin{equation}\label{lambdac}
\lambda_c=\min_{l,n>0}{\fr{\gamma_{ln}}{\alpha_l}\sqrt{(\gamma_{ln}^4+\sigma_1\alpha_l^2+\sigma_0n^2\pi^2)}}=\min_{l>0}{\sqrt{\fr{(\alpha_l^2+\pi^2)[(\alpha_l^1+\pi^2)^2+\sigma_1\alpha_l^2+\sigma_0\pi^2)]}{\alpha_l^2}}}.
\end{equation}
By the analysis for the convection scales in \cite{ptd}, we consider the case where the turbulent friction coefficients satisfy
\begin{equation}
1 \ll \sigma_0 \ll \sigma_1.
\end{equation}
In this case, the minimum in (\ref{lambdac}) is approximately achieved at
\begin{equation}\label{e65}
\left[\frac{h}{a}\right]^2 l_c(l_c+1)=\alpha_l^2 \sim \frac{\pi^2}{2} \left[\frac{\sigma_0}{\sigma_1}\right]^{1/2}.
\end{equation}
Then the same results as Theorems \ref{t4.1}, \ref{t2.2} and \ref{t2.3} hold true for the Rayleigh-B\'enard convection system (\ref{e3}) with added turbulent friction terms. We omit the details. It is worth, however, to remark that (\ref{e65}) provides a precise pattern selection mechanism in terms of the aspect ratio $h/a$, and the ratio between the horizontal and vertical friction coefficients. For example, for the circulation, we can choose the parameters as follows:
$$
a=6.4\times 10^6m,\,\,\,\ h=10^4m,\,\,\,\,\ l_c=6
$$
Then,
$$
\sigma_0/\sigma_1\approx 10^{-8}.
$$
Here $l_c=6$ represents the 6 Walker circulation cells over the tropical, and is also consistent with the three meridional circulation cells-the Hadley cell, the midlatitude cell and the polar cell, in the Northern hemisphere.
\bibliographystyle{siam}

\end{document}